\documentclass[10pt,journal]{IEEEtran}
\usepackage{amsmath}
\usepackage{epsfig}
\usepackage{graphicx}
\usepackage{amsmath}
\usepackage{amsfonts}
\usepackage{latexsym}
\usepackage{amssymb}
\usepackage{cite}
\usepackage{array}

\newcommand{\bs}{\mathbf{s}}
\newcommand{\bn}{\mathbf{n}}

\newcommand{\bx}{\mathbf{x}}
\newcommand{\bc}{\mathbf{c}}
\newcommand{\by}{\mathbf{y}}

\newcommand{\bO}{\mathbf{0}}

\newcommand{\bB}{\mathbf{B}}

\newcommand{\bH}{\mathbf{H}}
\newcommand{\bE}{\mathbf{E}}

\newcommand{\bW}{\mathbf{W}}

\newcommand{\bQ}{\mathbf{Q}}

\newcommand{\bC}{\mathbf{C}}

\newcommand{\bI}{\mathbf{I}}

\newcommand{\bh}{\mathbf{h}}

\newtheorem{theorem}{Theorem}
\begin{document}

\title{Efficient Soft-Input Soft-Output Detection of Dual-Layer MIMO Systems}
\author{Ahmad Gomaa and Louay M.A. Jalloul,~\IEEEmembership{Senior~Member,~IEEE}
\thanks{The authors are with the Mobile and Wireless Group, Broadcom Corporation, Sunnyvale, CA 94086 U.S.A.}}
\maketitle

\begin{abstract}
A dual-layer multiple-input multiple-output (MIMO) system with multi-level modulation is considered. A computationally efficient soft-input soft-output receiver based on the \textit{exact} max-log maximum a posteriori (max-log-MAP) principle is presented in the context of iterative detection and decoding. We show that the computational complexity of our exact max-log-MAP solution grows linearly with the constellation size and is also less than that of the best known methods of Turbo-LORD that only provide approximate solutions. Using decoder feedback to change the decision thresholds of the constellation symbols, we show that the exhaustive search operation boils down to a simple slicing operation.
\end{abstract}

\setlength{\textfloatsep}{10pt}
\setlength{\floatsep}{8pt}

\section{Introduction}\label{sec_Introduction}
Iterative detection and decoding (IDD) techniques have been widely used \cite{Tuchler02,Haykin02,SIOF2010,Mikami09} to improve the performance of multiple-input multiple-output (MIMO) systems. The detector utilizes the feedback from the decoder to enhance the accuracy of its output statistics. In \cite{Tuchler02,Haykin02}, the detector was designed as a linear minimum mean square error equalizer, accepting soft input from the channel decoder. The soft input was used to cancel the interference from other streams and to adapt the equalization (weight) vector by modifying the variance of the canceled streams. In \cite{SIOF2010}, the detector was designed as a decision feedback equalizer with successive cancellation at the symbol level before passing the log-likelihood ratios (LLRs) of the code bits to the decoder. In \cite{Mikami09}, IDD was used to mitigate the effect of inter-cell interference in orthogonal frequency division multiplexing (OFDM) systems. In \cite{TLORD_TWC,T-LORD,K-best}, a maximum a posteriori (MAP) approximating algorithm was proposed as an improvement over the layered orthogonal lattice detector (LORD) approach \cite{LORD-Alerton,LORD}. In \cite{Junti-IDD}, list detectors were proposed in addition to iterative channel estimation in OFDM systems. Other MAP approximation algorithms were proposed in \cite{MAP-approx-Choi,MAP-approx-Hassibi,Tellambura12,Hochwald03} where modified sphere detection techniques were used.

Dual-layer transmission schemes are widely used in current cellular systems where user equipments cannot easily support more than two antennas. The solution presented in this paper is an \textit{exact} solution of the max-log MAP detector for dual-layer systems and uses fewer metric computations than the approximate solution provided in \cite{TLORD_TWC,T-LORD}. To generate the LLRs for one layer, we use the a-priori LLRs generated by the turbo decoder for the other layer to modify its decision thresholds and then use the slicer as a simple search device.

The rest of the paper is organized as follows. The system model is described in Section \ref{sec_SysModel}, and the exact max-log MAP solution is derived in Section \ref{sec_exact_MaxLogMAP}. In Section \ref{sec_no_decesion_regions}, we prove that the a-priori probabilities can lead to constellation symbols with empty decision regions. In Section \ref{sec_algo}, we provide the complete algorithm and describe it in pseudo code. We analyze the algorithm computational complexity and compare its complexity with other algorithms in Section \ref{sec_complexity}, and the paper is concluded in Section \ref{sec_Conclusion}.

\textit{Notations}: Unless otherwise stated, lower case and upper case bold letters denote vectors and matrices, respectively, and $\bI_m$ denotes the identity matrix of size $m$. Furthermore, $\left|\;\right|$ and $\|\;\|$ denote the absolute value and the $l_2$-norm, respectively, while $(\,)^H$ denotes the complex conjugate transpose operation.

\section{System Model}\label{sec_SysModel}
We consider dual-layer transmission schemes, where two layers (streams) are transmitted over $N_t\geq 2$ antennas using the precoding matrix $\bW$ of size $N_t\times 2$. The receiver detects the transmitted streams using $N_r \geq 2$ receive antennas. The input-output relation is given by
\begin{equation}\label{eqn_I/P_O/P}
\by = \bar{\bH} \bW \bs +\bn \triangleq \bH\bs+\bn= \bh_1 s_1 + \bh_2 s_2 +\bn
\end{equation}
where $\by$, $\bs$, $\bn$ and $\bar{\bH}$ denote the $N_r\times 1$ received signal, $2\times 1$ transmitted symbols, $N_r \times 1$ background noise plus inter-cell interference, and $N_r \times N_t$ channel matrix, respectively. Furthermore, $\bh_i$ is the $i$-th column vector of the equivalent channel matrix $\bH=\bar{\bH}\bW$, and $s_i$ is the $i$-th transmitted symbol chosen from the $M$-QAM constellation $\chi$. The $i\text{-th}$ $M$-QAM symbol, $s_i$, represents $q=\text{log}_2(M)$ code bits $\bc_i=\begin{bmatrix}c_{i1}&c_{i2}&\ldots& c_{iq}\end{bmatrix}$. The above model suits single-carrier systems over flat fading channels and OFDM systems over frequency-selective channels where the relation in \eqref{eqn_I/P_O/P} applies to every subcarrier. In IDD, the detector computes the LLRs of the code bits and passes them to the channel decoder, which computes the extrinsic LLRs and feeds them back to the detector. The detector uses the a priori LLRs computed by the decoder to generate more accurate LLRs for the channel decoder and so forth. Assuming known channel and zero-mean circularly symmetric complex Gaussian noise $\bn$ of covariance matrix $\bC_{\bn\bn}=\bQ^{-1}$, we write the log MAP a posteriori detector LLR of the bit $c_{1k}$ as follows:
\begin{align}
&L(c_{1k})=\text{log}\nonumber\\
&\left(\frac
{\sum\limits_{\bar{s}_1\in \chi_{k,1}} \hspace{-0.1cm}P_{\bar{s}_1}\sum\limits_{\bar{s}_2\in \chi} \text{exp}\left(-\|\by-\bh_1 s_1 -\bh_2 s_2 \|_{\bQ}^2\right)\hspace{-0.05cm}P_{\bar{s}_2}}
{\sum\limits_{\bar{s}_1\in \chi_{k,0}} \hspace{-0.1cm}P_{\bar{s}_1}\sum\limits_{\bar{s}_2\in \chi} \text{exp}\left(-\|\by-\bh_1 s_1 -\bh_2 s_2 \|_{\bQ}^2\right)\hspace{-0.05cm}P_{\bar{s}_2}}\right) \label{eqn_logmap_LLR}
\end{align}
where $\|\bx \|_\bQ^2\equiv\bx^H\bQ\bx$, $\chi_{k,1}$ and $\chi_{k,1}$ denote the constellation sets where the $k$-th bit are '1' and '0', respectively, and $P_{\bar{s}_1}$ and $P_{\bar{s}_2}$ denote the a priori probabilities that $s_1=\bar{s}_1$ and $s_2=\bar{s}_2$, respectively. The max-log MAP approximation of the LLR of $c_{1k}$ is given by:
\begin{align}
&L(c_{1k})\nonumber\\
&\hspace{-0.1cm}=\hspace{-0.1cm}\max\limits_{\bar{s}_1\in \chi_{k,1}} \hspace{-0.15cm}\left(\hspace{-0.03cm}\text{log} P_{\bar{s}_1}\hspace{-0.1cm}+\max\limits_{\bar{s}_2\in \chi} \left(\text{log} P_{\bar{s}_2}\hspace{-0.1cm}-\hspace{-0.1cm}\|\by\hspace{-0.05cm}-\hspace{-0.05cm}\bh_1 \bar{s}_1 \hspace{-0.05cm}-\hspace{-0.05cm}\bh_2 \bar{s}_2 \|_{\bQ}^2\right)\hspace{-0.1cm}\right)\nonumber\\
&\hspace{-0.1cm}-\hspace{-0.1cm}\max\limits_{\bar{s}_1\in \chi_{k,0}} \hspace{-0.15cm}\left(\hspace{-0.03cm}\text{log} P_{\bar{s}_1}\hspace{-0.1cm}+\max\limits_{\bar{s}_2\in \chi} \left(\text{log} P_{\bar{s}_2}\hspace{-0.1cm}-\hspace{-0.1cm}\|\by\hspace{-0.05cm}-\hspace{-0.05cm}\bh_1 \bar{s}_1 \hspace{-0.05cm}-\hspace{-0.05cm}\bh_2 \bar{s}_2 \|_{\bQ}^2\right)\hspace{-0.1cm}\right)\label{eqn_LLR_s1_MaxLog}
\end{align}
Similarly, the max-log MAP LLR of $c_{2k}$ is:
\begin{align}
&L(c_{2k})\nonumber\\
&=\hspace{-0.1cm}\max\limits_{\bar{s}_2\in \chi_{k,1}} \hspace{-0.1cm}\left(\hspace{-0.03cm}\text{log} P_{\bar{s}_2}\hspace{-0.1cm}+\max\limits_{\bar{s}_1\in \chi} \left(\text{log} P_{\bar{s}_1}\hspace{-0.1cm}-\hspace{-0.1cm}\|\by\hspace{-0.05cm}-\hspace{-0.05cm}\bh_1 \bar{s}_1 \hspace{-0.05cm}-\hspace{-0.05cm}\bh_2 \bar{s}_2 \|_{\bQ}^2\right)\hspace{-0.1cm}\right)\nonumber\\
&-\hspace{-0.1cm}\max\limits_{\bar{s}_2\in \chi_{k,0}} \hspace{-0.1cm}\left(\hspace{-0.03cm}\text{log} P_{\bar{s}_2}\hspace{-0.1cm}+\max\limits_{\bar{s}_1\in \chi} \left(\text{log} P_{\bar{s}_1}\hspace{-0.1cm}-\hspace{-0.1cm}\|\by\hspace{-0.05cm}-\hspace{-0.05cm}\bh_1 \bar{s}_1 \hspace{-0.05cm}-\hspace{-0.05cm}\bh_2 \bar{s}_2 \|_{\bQ}^2\right)\hspace{-0.1cm}\right)\label{eqn_LLR_s2_MaxLog}
\end{align}
The brute force solution of \eqref{eqn_LLR_s1_MaxLog} (and similarly \eqref{eqn_LLR_s2_MaxLog}) requires the computation of $M^2$ metrics where, for each instance of $\bar{s}_1$, the metric $\left(\text{log} P_{\bar{s}_2}\hspace{-0.1cm}-\hspace{-0.1cm}\|\by-\bh_1 \bar{s}_1 -\bh_2 \bar{s}_2 \|_{\bQ}^2\right)$ is computed for all instances of $\bar{s}_2$. However, we show in Section \ref{sec_exact_MaxLogMAP} how we obtain the \textit{exact} max-log MAP solution for the LLRs using fewer than $4M$ (rather than $M^2$) metrics computations.

\section{Exact Max-Log MAP solution}\label{sec_exact_MaxLogMAP}
The main strategy of our solution is to convert the $l_2$-norms in \eqref{eqn_LLR_s1_MaxLog} and \eqref{eqn_LLR_s2_MaxLog} into simple absolute values fitted for the slicing operations. Then, we exploit the a-priori LLRs to control the the thresholds of the slicers. We begin by whitening the noise to get $\tilde{\by}$=$\sqrt{\bQ}\by$ and $\tilde{\bh}_i$=$\sqrt{\bQ}\bh_i$. We then rewrite the bottleneck maximization problem $\max\limits_{\bar{s}_2\in \chi}\left(\text{log} P_{\bar{s}_2}\hspace{-0.1cm}-\hspace{-0.1cm}\|\by-\bh_1 \bar{s}_1 -\bh_2 \bar{s}_2 \|_{\bQ}^2\right)$ as follows:
\begin{align}
&\max\limits_{\bar{s}_2\in \chi}\left(\text{log} P_{\bar{s}_2}-\|\tilde{\by}-\tilde{\bh}_1 \bar{s}_1 -\tilde{\bh}_2 \bar{s}_2 \|^2\right)\nonumber\\
=&\max\limits_{\bar{s}_2\in \chi}\left(\text{log} P_{\bar{s}_2}-\|\tilde{\bh}_2\|^2\left\|\frac{\left(\tilde{\by}-\tilde{\bh}_1 \bar{s}_1\right)}{\|\tilde{\bh}_2\|}
 -\bB \begin{bmatrix}  \bar{s}_2\\\bO \end{bmatrix} \right\|^2\right)
\end{align}
where $\bO$ is all-zero vector of length $N_r-1$, $\bB=\begin{bmatrix} \frac{\tilde{\bh}_2}{\|\tilde{\bh}_2\|} & \bE \end{bmatrix}$ is an $N_r\times N_r$ unitary matrix as $\bE$ is an $N_r \times (N_r-1)$ matrix chosen such that $\bE^H\tilde{\bh}_2=0$ and $\bE^H\bE=\bI$. The reason we write the matrix $\bB$ in this form is to exploit its unitary structure and take it as a common factor out of the norm without affecting its value. This will lead to converting the $l_2$-norm into a single absolute value as follows. Since $\bB^H\bB=\bI_2$, we rewrite the maximization as follows:
\begin{align}
&\max\limits_{\bar{s}_2\in \chi}\left(\text{log} P_{\bar{s}_2}-\|\tilde{\bh}_2\|^2\left\|\frac{\bB^H\left(\tilde{\by}-\tilde{\bh}_1 \bar{s}_1\right)}{\|\tilde{\bh}_2\|}
 - \begin{bmatrix}  \bar{s}_2\\\bO \end{bmatrix} \right\|^2\right) \nonumber\\
 =&\max\limits_{\bar{s}_2\in \chi}\left(\frac{\text{log} P_{\bar{s}_2}}{\|\tilde{\bh}_2\|^2}-\left| Z_{(\bar{s}_1)}
 - \bar{s}_2 \right|^2\right)\label{eqn_max_over_s2}
\end{align}
where $Z_{(\bar{s}_1)}=\tilde{\bh}_2^H\left(\tilde{\by}-\tilde{\bh}_1 \bar{s}_1\right)/\|\tilde{\bh}_2\|^2$. If the a-priori probability term ($\text{log} P_{\bar{s}_2}/\|\tilde{\bh}_2\|^2$) were not there \cite{French_paper_ML} (i.e., ML instead of MAP), then the solution of the maximization in \eqref{eqn_max_over_s2} would be a simple slicer, and only $2M$ metrics (enumeration over $\bar{s}_1$ and $\bar{s}_2$ in \eqref{eqn_LLR_s1_MaxLog} and \eqref{eqn_LLR_s2_MaxLog}) were to be computed to obtain the LLRs of the code bits corresponding to $s_1$ and $s_2$. With the a-priori probability term, we obtain the \textit{exact} solution of \eqref{eqn_max_over_s2} with a reasonable increase in the number of metrics computations which is, interestingly, less than that of the approximate solution in \cite{TLORD_TWC,T-LORD}. In modern communications standards \cite{LTE-211}, the real and imaginary parts of $\bar{s}_2$ correspond to two orthogonal $L$-PAM constellations, $\psi$, where $L=\sqrt{M}$\footnote{Assuming square constellation, without loss of generality.}. In Fig. \ref{fig_constellation}, we show the $L$-PAM one-dimensional constellation corresponding to the real or imaginary of any complex QAM constellation. Hence, we rewrite \eqref{eqn_max_over_s2} as follows:
\begin{figure}
\centering
\epsfig{file=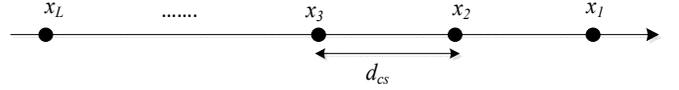,height=0.5in,width=3.45in}
\caption{$L$-PAM constellation for the real (or imaginary) part of the $M$-QAM constellation.} \label{fig_constellation}
\end{figure}
\begin{align}
&\max\limits_{\bar{s}_2\in \chi}\left(\frac{\text{log} P_{\bar{s}_2}}{\|\tilde{\bh}_2\|^2}-\left| Z_{(\bar{s}_1)}
 - \bar{s}_2 \right|^2\right)\nonumber\\
 &=\max\limits_{\bar{s}_{2,r}\in \psi}\left(\frac{\text{log} P_{\bar{s}_{2},r}}{\|\tilde{\bh}_2\|^2}-\left| Z_{\bar{s}_1,r}
 - \bar{s}_{2,r} \right|^2\right) \nonumber\\
 &+\max\limits_{\bar{s}_{2,I}\in \psi}\left(\frac{\text{log} P_{\bar{s}_{2},I}}{\|\tilde{\bh}_2\|^2}-\left| Z_{\bar{s}_1,I}
 - \bar{s}_{2,I} \right|^2\right)\label{eqn_max_s2_split}
\end{align}
where $Z_{\bar{s}_1,r}$ and $Z_{\bar{s}_1,I}$ denote the real and imaginary parts of $Z_{(\bar{s}_1)}$, respectively. Furthermore, $P_{\bar{s}_{2},r}$ and $P_{\bar{s}_{2},I}$ denote the a-priori probabilities that the real and imaginary parts of $s_2$ equal $\bar{s}_{2,r}$ and $\bar{s}_{2,I}$, computed using the a-priori LLRs of the bits corresponding to the real and imaginary parts, respectively.

Next, we use the a-priori probabilities (LLRs) to modify the decision regions of the $L$-PAM real and imaginary symbols; correspondingly apply the slicer to the real and imaginary parts of $Z_{(\bar{s}_1)}=\tilde{\bh}_2^H\left(\tilde{\by}-\tilde{\bh}_1 \bar{s}_1\right)/\|\tilde{\bh}_2\|^2$, respectively, to find the solution of \eqref{eqn_max_s2_split}; and then compute the metrics in \eqref{eqn_LLR_s1_MaxLog} and \eqref{eqn_LLR_s2_MaxLog}, which can be significantly simplified using \eqref{eqn_max_s2_split}. To develop the method of modifying the decision boundaries, we derive the decision region of the symbol $x_1$ in Fig. \ref{fig_constellation} by writing the conditions on $Z_{\bar{s}_1,r}$ such that
\begin{align}
\frac{\text{log} P_{x_1}}{\|\tilde{\bh}_2\|^2}-\left| Z_{\bar{s}_1,r}
 \hspace{-0.1cm}-\hspace{-0.1cm} x_1 \right|^2 > \frac{\text{log} P_{x_j}}{\|\tilde{\bh}_2\|^2}-\left| Z_{\bar{s}_1,r}
 \hspace{-0.1cm}-\hspace{-0.1cm} x_j \right|^2, \forall j\neq 1 \label{eqn_cond_x1}
\end{align}
Simplifying \eqref{eqn_cond_x1}, we get the decision region of $x_1$ as follows:
\begin{align}
Z_{\bar{s}_1,r} > \max\limits_{j>1} \left(\frac{x_1+x_j}{2}-\frac{\text{log} \left(P_{x_1}/P_{x_j}\right)}{2\left(x_1-x_j\right)\|\tilde{\bh}_2\|^2}\right)
\end{align}
Similarly, the decision region of $x_k$ is given by
\begin{align}\label{eqn_DecReg_xk}
\max\limits_{j>k} D_{kj} <  Z_{\bar{s}_1,r} < \min\limits_{j<k} D_{kj}, \qquad  1 < k < L
\end{align}
and the decision region of the last symbol $x_L$ is given by
\begin{gather}
Z_{\bar{s}_1,r} < \min\limits_{j<L} D_{Lj}\\
\text{where}\quad D_{kj}=D_{jk}=\frac{x_k+x_j}{2}-\frac{\text{log} \left(P_{x_k}/P_{x_j}\right)}{2\left(x_k-x_j\right)\|\tilde{\bh}_2\|^2}\label{eqn_Dkj}
\end{gather}
is called the \textit{probabilistic boundary} between the constellation symbols $x_j$ and $x_k$. Equation \eqref{eqn_Dkj} shows that the boundary between two neighboring symbols moves towards the symbol with the lower a-priori probability, tending to shrink its decision region while extending that of the symbol with the higher a-priori probability. Equation \eqref{eqn_Dkj} also shows that without a-priori LLRs (i.e., $P_{x_k}=\frac{1}{L},\;\forall k$), the boundaries between symbols return to their original values (the average of constellation symbols amplitudes).

\section{Symbols With Empty Decision Regions}\label{sec_no_decesion_regions}
We prove that the a-priori probability distribution can lead to constellation symbols with empty decision regions that will not be chosen by the slicer regardless of $Z_{\bar{s}_1,r}$ (or $Z_{\bar{s}_1,I}$).

\begin{theorem}
When computing the lower bound of the decision region for the constellation symbol $x_k$, given by $\max_{j>k} D_{kj}$, the following can occur:
\begin{equation}
j^*> k+1, \quad \text{where}\quad j^*\equiv arg \max\limits_{j>k} D_{kj}
\end{equation}
meaning that the lower bound of the symbol $x_k$ is not determined by its boundary with the adjacent symbol $x_{k+1}$, but determined instead by its boundary with a farther symbol $x_{j^*} < x_{k+1}$. In this case, all symbols lying between $x_k$ and $x_{j^*}$ (i.e., the constellation symbols $x_m$, where $k<m<j^*$) do not have decision regions and will not be chosen regardless of the decision statistic value.
\end{theorem}

\begin{proof}
From \eqref{eqn_DecReg_xk}, the decision boundaries for $x_m$, where $k<m<j^*,$ are given by
\begin{align}
\max &\left( D_{m(m+1)},..,D_{mj^*},..,D_{mL} \right) < Z_{\bar{s}_1,r}\nonumber\\
&< \min \left( D_{m(m-1)},..,D_{mk},..,D_{m1} \right) \label{eqn_ineq_proof}
\end{align}
However, there is no value for $Z_{\bar{s}_1,r}$ that satisfies \eqref{eqn_ineq_proof} if
\begin{equation}\label{eqn_cond_proof}
D_{mk} < D_{mj^*}
\end{equation}
In the sequel, we prove that the condition in \eqref{eqn_cond_proof} is satisfied if  $j^*=arg \max\limits_{j>k} D_{kj}$, i.e., $D_{kj^*} > D_{km}$ and, hence,
\begin{gather}
x_{j^*} \hspace{-0.05cm}-\hspace{-0.05cm} x_m \hspace{-0.05cm}>\hspace{-0.05cm}
 \frac{\text{log} \left(P_{x_k}/P_{x_{j^*}}\right)}{2\left(x_k-x_{j^*}\right)\|\tilde{\bh}_2\|^2}
 - \frac{\text{log} \left(P_{x_k}/P_{x_m}\right)}{2\left(x_k-x_m\right)\|\tilde{\bh}_2\|^2} \\
 (m \hspace{-0.1cm}-\hspace{-0.1cm} j^*) \,d_{cs} > \frac{\text{log} \left(\left(P_{x_k}/P_{x_{j^*}}\right)^{m\hspace{-0.05cm}-k}\left(P_{x_m}/P_{x_k}\right)^{j^*\hspace{-0.1cm}-k}\right)}
 {(j^*-k)(m-k)d_{cs}\|\tilde{\bh}_2\|^2}\label{eqn_inequality_basic}
\end{gather}
where $x_k - x_{j^*} = (j^* - k)\,d_{cs}$ and $d_{cs}$ is the separation between adjacent real (or imaginary) constellation symbols as shown in Fig. \ref{fig_constellation}. Since $k<m<j^*$, we define
\begin{align}\label{eqn_definitions}
m = k + f,\qquad j^* = m + g = k + f + g
\end{align}
where $f,g\in\{0,\mathbb{Z}^+\}$. We rewrite \eqref{eqn_inequality_basic} as follows:
\begin{equation}\label{eqn_inequality_2}
fg(f + g)d_{cs}^2 \|\tilde{\bh}_2\|^2 \hspace{-0.05cm}<\hspace{-0.05cm} \text{log} \left(\left(P_{x_{j^*}}/P_{x_m}\right)^f\left(P_{x_k}/P_{x_m}\right)^g\right)
\end{equation}
Next, we rewrite the condition in \eqref{eqn_cond_proof} as follows:
\begin{align}
D_{mk} \hspace{-0.05cm}-\hspace{-0.05cm} D_{mj^*} \hspace{-0.05cm}=\hspace{-0.05cm} \frac{(f\hspace{-0.05cm}+\hspace{-0.05cm}g)\,d_{cs}}{2}-\hspace{-0.05cm}
\frac{\text{log}\hspace{-0.05cm} \left(\hspace{-0.05cm}\left(\frac{P_{x_k}}{P_{x_m}}\right)^g\left(\frac{P_{x_{j^*}}}{P_{x_m}}\right)^f\right)}
 {2fg\,d_{cs}\|\tilde{\bh}_2\|^2}
\end{align}
Using the inequality in \eqref{eqn_inequality_2}, we bound $D_{mk} \hspace{-0.05cm}-\hspace{-0.05cm} D_{mj^*}$ as follows:
\begin{align}
D_{mk} - D_{mj^*} &< \frac{(f+g)\,d_{cs}}{2} - \frac{fg(f+g)d_{cs}^2\|\tilde{\bh}_2\|^2}{2fg\,d_{cs}\|\tilde{\bh}_2\|^2}\nonumber\\
D_{mk} - D_{mj^*} &<0,\qquad D_{mk} < D_{mj^*}
\end{align}
which concludes the proof.
\end{proof}
The practical importance of this theorem is that it can reduce the algorithm complexity and further speed it up. For example, if the lower boundary of $x_1$ is determined by $x_4$ then we do not need to compute the decision boundaries of $x_2$ and $x_3$ because they will have empty decision regions.

\section{Algorithm and Computational Complexity}\label{sec_algo}
In the sequel, we summarize the algorithm and show the receiver model in Fig. \ref{fig_system_model}.
\\
{\bf Preprocessing}: Compute $\bH=\bar{\bH} \bW$ and whiten the noise by computing $\tilde{\by}=\sqrt{\bQ}\by$, $\tilde{\bh}_1=\sqrt{\bQ}\bh_1$, and $\tilde{\bh}_2=\sqrt{\bQ}\bh_2$.
\\
{\bf Procedure}:
\begin{enumerate}
\item\label{item_1} Get the decision regions for $Z_{\bar{s}_1,r}, Z_{\bar{s}_1,I}, Z_{\bar{s}_2,r},$ and $Z_{\bar{s}_2,I}$ using the corresponding a priori LLRs as follows:
\\
Initialize $k=1$.
\\
{\bf While} $k<=L$
\\
A) Compute the lower and upper thresholds of the $k$-th constellation symbol as $\max_{j>k} D_{kj}$ and $\min_{j<k} D_{kj}$, respectively, where
\begin{equation}\label{eqn_Dkj_LLRs}
D_{kj}=\frac{x_k+x_j}{2}-\frac{\sum_{n=1}^{q/2} \left((b_{n,k}-b_{n,j})L_a(c_{i\,n})\right)}{2\left(x_k-x_j\right)\|\tilde{\bh}_i\|^2}
\end{equation}
where $i\in\{1,2\}$, $L_a(c_{i\,n})$ denotes the a priori LLR of the code bit $c_{in}$, and $\left\{b_{n,k},b_{n,j}\right\}_{n=1}^{q/2}\in \left\{0,1\right\}$ are the bit vectors corresponding to the constellation symbols $x_k$ and $x_j$, respectively. The transition from the probability domain in \eqref{eqn_Dkj} to the LLR domain in \eqref{eqn_Dkj_LLRs} is straightforward.

B) {\bf If} $j^*> k+1, \quad \text{where}\quad j^*\equiv arg \max\limits_{j>k} D_{kj}$, set the decision regions of the symbols $x_m$, where $k<m<j^*$, to empty, and set $k=j^*.$
\\
\hspace*{0.34cm} {\bf Else}, set $k=k+1.$
\\
{\bf End While}
\item Enumeration step over constellation points of $s_1$ and $s_2$.
\\
{\bf For} $k=1:M$
\\
a) Compute the following quantities for $\bar{s}_1(k),\bar{s}_2(k) \in \chi$
\begin{equation}\nonumber
Z_{(\bar{s}_1\hspace{-0.02cm})}\hspace{-0.1cm}=\hspace{-0.1cm}\frac{\tilde{\bh}_2^H}{\|\tilde{\bh}_2\|^2}\left(\hspace{-0.05cm}\tilde{\by}\hspace{-0.1cm}-\hspace{-0.1cm}\tilde{\bh}_1 \bar{s}_1(k)\hspace{-0.05cm}\right)\hspace{-0.1cm},
Z_{(\bar{s}_2\hspace{-0.02cm})}\hspace{-0.1cm}=\hspace{-0.1cm}\frac{\tilde{\bh}_1^H}{\|\tilde{\bh}_1\|^2}\left(\hspace{-0.05cm}\tilde{\by}\hspace{-0.1cm}-\hspace{-0.1cm}\tilde{\bh}_2 \bar{s}_2(k)\hspace{-0.05cm}\right)
\end{equation}
b) Slice the real and imaginary parts of $Z_{(\bar{s}_1)}$ and $Z_{(\bar{s}_2)}$ using the thresholds obtained in Step \ref{item_1} to obtain $\bar{s}^*_2(k)$ and $\bar{s}^*_1(k)$, respectively.
\\
c) Compute the following metrics
\vspace{-0.2cm}
\begin{align}
\eta_1(k)&=\sum_{n=1}^q \left(\bar{b}_{1,n}(k)L_a(c_{1n})+\bar{b}^*_{2,n}(k)L_a(c_{2n})\right)\nonumber\\
&-\|\tilde{\by}-\tilde{\bh}_1 \bar{s}_1(k) -\tilde{\bh}_2 \bar{s}^*_2(k) \|^2\label{eqn_metric_eta1}\\
\eta_2(k)&=\sum_{n=1}^q \left(\bar{b}^*_{1,n}(k)L_a(c_{1n})+\bar{b}_{2,n}(k)L_a(c_{2n})\right)\nonumber\\
&-\|\tilde{\by}-\tilde{\bh}_1 \bar{s}^*_1(k) -\tilde{\bh}_2 \bar{s}_2(k) \|^2
\end{align}
where $\{\bar{b}_{1,n}(k),\bar{b}_{2,n}(k),\bar{b}^*_{1,n}(k),\hspace{-0.05cm}\bar{b}^*_{2,n}(k)\hspace{-0.05cm}\}_{n=1}^{q}$ are the bit vectors of $\bar{s}_1(k), \bar{s}_2(k), \bar{s}_1^*(k), \bar{s}^*_2(k)$, respectively.
\\
{\bf End For}
\item Compute the detector LLRs for $i=1,2$ and $0\leq n\leq q$
\begin{align}
L(c_{in})=\max_{k\,:\,\bar{b}_{i,n}(k)=1} \eta_i(k) - \max_{k\,:\,\bar{b}_{i,n}(k)=0} \eta_i(k)
\end{align}
\end{enumerate}

\begin{figure}
\centering
\epsfig{file=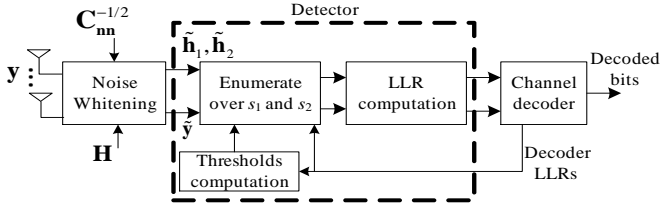,height=1.1in,width=3.45in}
\caption{Receiver model with our proposed approach.} \label{fig_system_model}
\end{figure}

\section{Complexity analysis}\label{sec_complexity}
We count the number of required metrics computations to obtain the $2\text{log}_2(M)$ detector LLRs corresponding to $s_1$ and $s_2$. To get the new decision regions, we need to compute the probabilistic boundaries between every two symbols of the $L$ symbols (for both real and imaginary parts). Since $D_{jk}=D_{kj}$, the number of metrics (boundaries) to be computed is
\begin{equation}
N_{metrics}^{s_1} = 2\frac{L(L-1)}{2} = L^2-L = M-\sqrt{M} < M
\end{equation}
Note that these boundaries are computed only once and are not included inside the enumeration over $s_1$ in \eqref{eqn_LLR_s1_MaxLog}. Hence, the total number of metric computations to obtain the LLRs of the code bits corresponding to $s_1$ is $M+N_{metrics}^{s_1}=2M-\sqrt{M}$. To obtain the $2q$ LLRs corresponding to $\bs$ (i.e., $s_1$ and $s_2$), the number of metric computations per tone becomes
\begin{equation}
N_{metrics}^{Total}= 4M-2\sqrt{M} < 4M
\end{equation}

In Table \ref{Table_cmplx_comp}, we compare our algorithm with the Turbo-LORD (T-LORD) \cite{TLORD_TWC} and the brute-force algorithms in terms of the number of metrics to be computed, number of real multiplications (Muls), and number of real additions (Adds) per tone per iteration as function of the constellation size and the number of receive antennas. In Table \ref{Table_cmplx_comp_256QAM}, we compare these algorithms for 256-QAM and two receive antennas where we observe the significant computational complexity saving without any performance loss since our algorithm obtains the exact solution rather than the approximate solution in \cite{TLORD_TWC}. In T-LORD \cite{TLORD_TWC}, while enumerating over $s_1$, three candidates for $s_2$ are obtained for every possible value of the $M$ candidiases of $s_1$. Hence, we have $3M$ candidates for the $\left(s_1,s_2\right)$ pair, and the metric in \eqref{eqn_metric_eta1} is computed for each candidate. Doing the same for $s_2$, we have another $3M$ metrics summing up to $6M$ metrics to be computed.


\begin{table}
\caption{Complexity comparison between various ML detectors}
\centering
\begin{tabular}{|p{1.5cm}|p{0.8cm}|p{2cm}|p{2.6cm}|}
\hline
Detector & Metrics & Real Muls & Real Adds \\ \hline
Proposed & $4M-2\sqrt{M}$ & $\left(16N_r+18\right)M-2\sqrt{M}$  &  $\left(12N_r+3q+18\right)M-\left(q+2\right)\sqrt{M}-4$  \\ \hline
T-LORD & $6M$   & $\left(16N_r+48\right)M $  &  $\left(12N_r+2q+52\right)M-4$ \\ \hline
Brute force & $M^2$ & $8M^2$ & $12M^2-4M$ \\ \hline
\end{tabular}
\label{Table_cmplx_comp}
\end{table}

\begin{table}
\caption{Complexity of various ML detectors for $M=$256 and $N_r=2$}
\centering
\begin{tabular}{|c|c|c|c|}
\hline
Detector & Metrics & Real Muls & Real Adds \\ \hline
Proposed & 992 & 12768  &  16732 \\ \hline
T-LORD & 1536   & 20480  &  23548 \\ \hline
Brute force & 65536 & 524288 & 785408 \\ \hline
\end{tabular}
\label{Table_cmplx_comp_256QAM}
\end{table}

\section{Conclusion}\label{sec_Conclusion}
We developed the \textit{exact} max-log MAP detector for IDD in dual-layer MIMO schemes with computational complexity less than $4M$. The idea is to use the a priori LLRs in modifying the decision thresholds of the constellation symbols. We also showed that the a priori LLRs can lead to constellation symbols with empty decision regions, reducing the search space of the slicing block. Comparing the computational complexity with the Turbo-LORD approximate solution and the exact brute force solutions, we show that our algorithm achieves significant complexity reduction while achieving the exact max-log MAP solution. We have numerically verified that our method yields the same performance as the brute force solution for various simulation parameters but the simulation results are not shown here due to space limitations.



\end{document}